\documentclass[conference,letterpaper]{IEEEtran}
\IEEEoverridecommandlockouts
\usepackage{cite}
\usepackage{algorithmic}
\usepackage{graphicx}
\usepackage{textcomp}
\usepackage{xcolor}

\usepackage{amsmath,amssymb,amsfonts,amsthm,mathtools,bm}
\usepackage{comment}

\newtheorem{theorem}{Theorem}
\newtheorem{proposition}{Proposition}
\newtheorem{lemma}{Lemma}
\newtheorem{corollary}{Corollary}

\theoremstyle{definition}

\newtheorem{definition}{Definition}
\newtheorem{assumption}{Assumption}

\theoremstyle{remark}
\newtheorem{remark}{Remark}

\newcommand{\bs}[1]{\boldsymbol{#1}}
\newcommand{\dd}{\mathrm{d}}
\newcommand{\ii}{\mathrm{i}}

\def\BibTeX{{\rm B\kern-.05em{\sc i\kern-.025em b}\kern-.08em
    T\kern-.1667em\lower.7ex\hbox{E}\kern-.125emX}}

\begin{document}

\title{Affine-Projection Recovery of Continuous Angular Power Spectrum: Geometry and Resolution
}

\author{
\IEEEauthorblockN{Shengsong Luo\IEEEauthorrefmark{1}, Chongbin Xu\IEEEauthorrefmark{1}, Xin Wang\IEEEauthorrefmark{1},
Ruilin Wu\IEEEauthorrefmark{2}, Junjie Ma\IEEEauthorrefmark{3}, Xiaojun Yuan\IEEEauthorrefmark{4}}
\IEEEauthorblockA{\IEEEauthorrefmark{1}Fudan University, Shanghai, China.
Email: {\{ssluo22@m., chbinxu@, xwang11@\}fudan.edu.cn}}
\IEEEauthorblockA{\IEEEauthorrefmark{2}Peking University, Beijing, China.
Email: rlwu22@stu.pku.edu.cn}
\IEEEauthorblockA{\IEEEauthorrefmark{3}Chinese Academy of Sciences, Beijing, China.
Email: majunjie@lsec.cc.ac.cn}
\IEEEauthorblockA{\IEEEauthorrefmark{4}University of Electronic Science and Technology of China, Chengdu, China.
Email: xjyuan@uestc.edu.cn}
}

\maketitle

\begin{abstract}
This paper considers recovering a continuous angular power spectrum (APS) from the channel covariance. Building on the projection-onto-linear-variety (PLV) algorithm, an affine-projection approach introduced by Miretti \emph{et. al.}, we analyze PLV in a well-defined \emph{weighted} Fourier-domain to emphasize its geometric interpretability. This yields an explicit fixed-dimensional trigonometric-polynomial representation and a closed-form solution via a positive-definite matrix, which directly implies uniqueness. 
We further establish an exact energy identity that yields the APS reconstruction error and leads to a sharp identifiability/resolution characterization: PLV achieves perfect recovery if and only if the ground-truth APS lies in the identified trigonometric-polynomial subspace; otherwise it returns the minimum-energy APS among all covariance-consistent spectra. 
\end{abstract}


\section{Introduction}
	We consider a directional (angle-domain) channel model as in~\cite{Miretti-2021}. 
Specifically, we study a single-cell uplink where the base station (BS) employs an $M$-antenna uniform linear array (ULA) and each user has a single antenna. 
A classical spatial channel covariance $\bs R \in \mathbb C^{M \times M}$ gives~~\cite{Miretti-2021,bjornson2016massive,yin2014dealing,haghighatshoar2016massive,besson2005matched}.
\begin{equation}
	\label{eq: R_representation}
	\bs{R}_{}
	=
	\int_{\Theta}
	\rho_{\star}(\theta)\,
	\bs{a}(\theta)
	\bs{a}(\theta)^H 
	\dd\theta,
\end{equation}
where $\theta\in \Theta:=[-\pi/2,\pi/2]$ is the angle of arrival (AoA), $\rho_{\star}:\Theta \rightarrow \mathbb R_{+}$ denotes the angular power spectrum (APS) satisfying Assumption~\ref{assump: APS}, and $\bs{a}: \Omega \rightarrow \mathbb C^M$ is the frequency dependent antenna array responses:
\begin{align}
	\label{eq:steering}
	\bs{a}(\theta)
	:=[
	1,
	e^{\ii \kappa_1 \sin\theta},
	\ldots,
	e^{\ii \kappa_{M-1}\sin\theta}
	]^\top,
\end{align}
where $\kappa_m := \gamma \pi m$, $m \in \{0,1,\ldots,M-1\}$, with $\gamma=2d/\lambda_{f}$. Here, $d$ is the inter-element spacing of the ULA, and $\lambda_{f}$ is the wavelength of carrier frequency $f$.

The APS $\rho_{\star}$ is the primary quantity we aim to estimate. 
It provides a compact, second-order description of the propagation geometry (e.g., dominant directions, angular spread, and diffuse multipath) and is useful for long-term system design tasks such as beamforming~\cite{mohammadzadeh2022covariance,shakya2025angular}, user grouping~\cite{haghighatshoar2016massive}, and channel prediction~\cite{fan2018angle,wang2024channel,Bameri-2023}. 
Moreover, $\rho_{\star}$ typically varies on a much slower timescale than instantaneous channel realizations, which makes it well suited for long-term inference~\cite{Haghighatshoar-2018,Barzegar-2019,Bameri-2023}.

The covariance in model~\eqref{eq: R_representation} encodes the underlying continuous angular scattering spectrum $\rho_{\star}$.
This perspective is closely related to covariance-based random access, which considers a discretized counterpart of~\eqref{eq: R_representation} with random signatures~\cite{haghighatshoar2018improved,chen2021phase,fengler2021non}. 
In that setting, it is shown that the covariance is a sufficient statistic for estimating large-scale parameters (e.g., user activity and large-scale fading coefficients)~\cite[Theorem~1]{haghighatshoar2018improved}, supporting the general principle that second-order statistics can capture the relevant long-term information.

Miretti \emph{et al.}~\cite{Miretti-2021} proposed an influential estimator---termed projection onto a linear variety (PLV)---for the inverse problem of recovering a continuous APS from a spatial covariance matrix. 
Given a target covariance $\bs R$, PLV selects, among all spectra consistent with the covariance model~\eqref{eq: R_representation}, the one with minimum $L^2$-norm (cf. Definition~\ref{def: weighted-L2}):
\begin{equation}
	\rho_{\mathrm{plv}}
	=
	\arg\min_{\rho \in \mathcal{V}} \ \|\rho\|_{},
\end{equation}
where $\mathcal{V}$ is an affine subset of the Hilbert space $L^2(\Theta)$ defined by the covariance constraints
\begin{equation}
	\mathcal{V}
	:=
	\big\{
	\rho \in L^2(\Theta) \mid
	\bs{R}
	=
	\int_{\Theta}
	\rho(\theta)\,
	\bs{a}(\theta)\bs{a}(\theta)^H 
	\,\mathrm{d}\theta
	\big\}.
\end{equation}
Equivalently, $\rho_{\mathrm{plv}}$ is the orthogonal projection of the origin onto the affine set $\mathcal{V}$ in $L^2(\Theta)$~\cite{Miretti-2021}, which motivates the term ``projection onto a linear variety" and our use of ``affine projection" in the title.

Building upon the PLV framework of~\cite{Miretti-2021}, we revisit PLV in a \emph{weighted} Fourier-domain and sharpen its geometric structure. In particular, we make the following three aspects explicit (\emph{our main contributions}):
\begin{itemize}
	\item \textbf{Appearance of trigonometric polynomials.}
	We introduce a weighted Fourier-domain perspective that reveals a clean geometric decomposition of the covariance-consistency set: The feasible spectra form an affine set, and its orthogonal complement is a fixed-dimensional subspace, in terms of \emph{trigonometric polynomials}, determined by the array size. This turns an infinite-dimensional recovery problem into an explicit finite-dimensional characterization.
	
	\item \textbf{Equivalent PLV interpretations and its uniqueness.}
	We show that the PLV estimator admits several equivalent interpretations, and we derive a simple closed-form procedure to compute the solution. This also yields a clear uniqueness guarantee and makes the dependence on the observed covariance lags transparent.
	
	\item \textbf{Exact error decomposition and resolution.}
	We establish an exact energy/error decomposition that directly quantifies the reconstruction error of PLV. This leads to a sharp identifiability (resolution) statement: Perfect recovery holds precisely for spectra that lie in the identified subspace; otherwise PLV returns the minimum-energy spectrum among all covariance-consistent candidates, clarifying the fundamental aperture-limited resolution.
\end{itemize}

\subsection{Preliminaries and Notation}
Let $I\subseteq\mathbb R$ be a Lebesgue-measurable set and $\mathrm dx$ be the Lebesgue measure. We introduce the following preliminaries.

\begin{definition}
	\label{def: weighted-L2}
	\textit{(Weighted $L^{2}$ space)}.
	Let $w: I\to(0,\infty)$ be a Lebesgue-measurable weight function with $w>0$ almost everywhere.
	For measurable functions $f,g:I\to\mathbb R$, define the weighted inner product
	\[
	\langle f, g \rangle_{w} := \int_{I} f(x)\, g(x)\, w(x)\, \mathrm{d}x,
	\]
	whenever it is well-defined.
	Define the weighted $L^{2}$ space
	\begin{align*}
		L^{2}_{w}( I)
		:= & \bigl\{\, f: I\to\mathbb R \ \big|\ \langle f,f\rangle_{w} < \infty \,\bigr\} ,
	\end{align*}
	with the inner product $\langle \cdot,\cdot\rangle_w$ and norm
	\(
	\|f\|_{w}:=\sqrt{\langle f, f\rangle_{w}}.
	\)
	If $w(x)= 1$ on $I$, then $L_w^2(I)$ reduces to the standard space $L^2(I)$ with the usual inner product and norm 
	\(
	\|f\|_{}:=\sqrt{\langle f, f\rangle_{}}.
	\)
\end{definition}

\begin{definition}
	\label{def: projection}
	\emph{(Projection Operator).}
	Let $\mathcal{S}\subset L^2_w$ be a nonempty closed convex set. The
	\emph{projection} of $f\in L^2_w$ onto $\mathcal{S}$ is the (unique)
	element $\mathbb{P}_{\mathcal{S}}(f)\in\mathcal{S}$ that solves
	\begin{equation}
		\mathbb{P}_{\mathcal{S}}(f)
		:= \arg\min_{z\in\mathcal{S}} \,\|f - z\|_{w}^2.
	\end{equation}
	The mapping $\mathbb{P}_{\mathcal{S}}:L^2_w\to L^2_w$ is called the
	\emph{projection operator} onto $\mathcal{S}$. 
	In particular, if $\mathcal{S}$ is a closed subspace, then $\mathbb{P}_{\mathcal{S}}$
	is the orthogonal projection onto $\mathcal{S}$.
\end{definition}

\begin{assumption}
	\label{assump: APS}
	The APS $\rho_{\star} \in L^2(\Theta)$.
\end{assumption}

\textit{Organization.}
The rest of this paper is organized as follows. Section~\ref{sec:system-model} presents the transformed (weighted Fourier) domain formulation and the associated weighted $L_2$ geometric setup. Section~\ref{sec: main-results} develops the main geometric results and Section~\ref{sec: conclusion} concludes the paper.

\emph{Notation}. 
Throughout the paper, we use the typefaces $a, \bs{a}, \bs{A}$ and $\mathcal{A}$ to denote scalar/function, vector, matrix and set, respectively.
For a vector $\bs{a} \in \mathbb{R}^n$,
$\bs{a}(j_1:j_2)$ denotes the subvector of $\bs{a}$ consisting of its
$j_1$-th to $j_2$-th entries, where $0 \le j_1 \le j_2 \le n-1$; in
particular, $\bs{a}[j]$ denotes the $j$-th entry of $\bs{a}$, and we
also write $a_j$ when no ambiguity can arise. For a matrix
$\bs{A} \in \mathbb{R}^{m \times n}$, $[\bs{A}]_{i,j}$ denotes its
$(i,j)$-th entry, and
$\bs{A}(i_1\!:\!i_2,\,j_1\!:\!j_2)$ denotes the submatrix formed by
rows $i_1$ to $i_2$ and columns $j_1$ to $j_2$, where
$0 \le i_1 \le i_2 \le m-1$ and $0 \le j_1 \le j_2 \le n-1$.

	\section{Weighted Fourier-Domain Formulation}
\label{sec:system-model}

In this section, we reformulate the continuous APS recovery problem in a \emph{weighted} Fourier-domain that is more amenable to geometric analysis.
For a ULA, the covariance matrix $\bs R$ is Hermitian Toeplitz and positive semidefinite, and is therefore fully specified by its first column $\bm r\in\mathbb C^{M}$:
\begin{equation}
	\label{eq: r-vec-exp}
	\bm{r}_{} := [\bs{R}_{}]_{:,1} = \int_{\Theta}
	\rho_{\star}(\theta)\,
	\bs{a}(\theta) \dd \theta.
\end{equation}
Let $\bm r=[r_0,r_1,\ldots,r_{M-1}]^\top$. Then, the $m$-th covariance lag admits the integral form
\begin{equation}
	\label{eq: r_m exp}
	r_m
	=
	\int_{\Theta} \rho_{\star}(\theta)\, e^{\ii \kappa_m \sin\theta}\,\dd \theta,
	\quad m=0,\ldots,M-1,
\end{equation}
where $\kappa_m$ is defined in~\eqref{eq:steering} (with $r_0=\int_\Theta \rho_\star(\theta)\dd\theta$).

We next apply the change of variables $x=\sin\theta$ and define the transformed APS
\begin{equation}
	\label{eq: g_def}
	g_{\star}(x):= \rho_{\star}(\arcsin x), \quad x\in I:=[-1,1].
\end{equation}
Then~\eqref{eq: r_m exp} becomes
\begin{align}
	\label{eq: integral_g}
	r_m
	&=
	\int_{-1}^{1} g_{\star}(x)\, e^{\ii \kappa_m x}\, w(x)\,\dd x
	\nonumber\\
	&=
	\langle g_{\star},\, e^{\ii\kappa_m(\cdot)}\rangle_{w},
	\qquad m=0,\ldots,M-1,
\end{align}
where the Jacobian weight is
\begin{equation}
	\label{eq: weight}
	w(x)=\frac{1}{\sqrt{1-x^{2}}},
\end{equation}
and $\langle\cdot,\cdot\rangle_w$ is the weighted inner product in Definition~\ref{def: weighted-L2}.

\begin{remark}
	\label{rem: L2_to_L2w}
	Under Assumption~\ref{assump: APS}, $\rho_{\star}\in L^2(\Theta)$ implies $g_{\star}\in L_w^2(I)$ via the mapping~\eqref{eq: g_def}.
\end{remark}

Since $\rho_{\star}$ and $g_{\star}$ are in one-to-one correspondence, we will refer to either function as the APS when no confusion arises.
In the remainder of the paper, we work with the \emph{weighted} Fourier-domain representation~\eqref{eq: integral_g}.
This viewpoint reveals that the available covariance lags $\{r_m\}_{m=0}^{M-1}$ are weighted Fourier measurements of $g_{\star}$, and it will lead to a cleaner geometric characterization of the PLV estimator.

\section{Geometric Analysis of PLV in the Weighted Fourier-Domain Representation}
\label{sec: main-results}

We next revisit the PLV scheme and place it in the \emph{weighted} Fourier-domain representation. This leads to a arguably simple geometric interpretation of PLV 
and shows that PLV reconstructs the APS within a fixed-order
trigonometric polynomial model. 

\subsection{Geometry and Trigonometric-Polynomial Subspace}

In order to facilitate geometric analysis, we introduce the following Lemma.

\begin{lemma}[Affine structure and orthogonal decomposition]
	\label{lem: Affine-sub-space}
	Let $\bs{r} = [r_0,\ldots,r_{M-1}]^{\top} \in \mathbb{C}^{M}$ be a
	given covariance vector, and define $r_{-m} := r_m^{*}$ for
	$m=1,\ldots,M-1$. Consider the $L^2_w$ space in Definition~\ref{def: weighted-L2}
	and define the affine set
	\begin{align}
		\mathcal{V}_w
		&:= \big\{ f  \;|\;
		\langle f, e^{\ii \kappa_m (\cdot)} \rangle_w = r_m,
		\; -M < m < M \big\}, \label{eq: Vr-def}
	\end{align}
	together with the closed subspace (i.e., null space)
	\begin{align}
		\mathcal{N}
		&:= \big\{ f  \;|\;
		\langle f, e^{\ii \kappa_m (\cdot)} \rangle_w = 0,
		\; -M < m < M \big\}, \label{eq:N-def}
	\end{align}
	and its orthogonal complement
	\begin{align}
		\mathcal{N}_{\bot}
		&:= \{ f  \;|\;
		\langle f, h \rangle_w = 0,\;\forall\, h \in \mathcal{N} \}.
		\label{eq:N-perp-def}
	\end{align}
	Then the following two statements hold true:
	
	1) $\mathcal{V}_w$ is an affine subset of $L^2_w(I)$
	with direction $\mathcal{N}$, i.e., for any
	$f_1 \in \mathcal{V}_w$ and
	$f_2 \in \mathcal{N}$ we have $f_1 + f_2 \in \mathcal{V}_w$,
	and hence
	\(
	\mathcal{V}_w + \mathcal{N} = \mathcal{V}_w.
	\)
	
	2) The orthogonal complement $\mathcal{N}_{\bot}$ is a
	finite-dimensional subspace spanned by trigonometric
	polynomials, i.e.,
	\begin{equation}
		\label{eq:complement-space}
		\mathcal{N}_{\bot}
		= \big\{ f(\cdot;\bs{b}) \;|\;
		\bs{b} \in \mathbb{R}^{2M-1} \big\},
	\end{equation}
	in which
	\[
	f(x;\bs{b})
	= b_0
	+ \sum_{m=1}^{M-1}
	b_m \cos(\kappa_m x)
	+ b_{M-1+m} \sin(\kappa_m x) 
	\]
	with $\kappa_m$ denoting the spatial-frequency samples
	associated with the covariance lags (cf.~\eqref{eq:steering}).
\end{lemma}
\begin{proof}
	See Appendix~\ref{proof: Affine}.
\end{proof}


Lemma~\ref{lem: Affine-sub-space} shows that the set of functions $g \in L^2_w(I)$ that satisfy the covariance constraints is an affine set $\mathcal{V}_{w}$ with direction $\mathcal{N}$, and that $\mathcal{N}_{\perp}$ is a finite-dimensional subspace spanned by trigonometric polynomials.
The PLV can
thus be viewed as selecting the unique element in
$\mathcal{V}_w \cap \mathcal{N}_{\bot}$ (see Fig.~\ref{fig:PLV}). Building
on this geometric picture, the following Theorem summarizes
several equivalent characterizations of PLV.

\begin{theorem}
	\label{prop:PLV-equivalent}
	Recall the sets $\mathcal{V}_w$, $\mathcal{N}_{\bot}$ defined in
	Lemma~\ref{lem: Affine-sub-space}. Let $g_{\mathrm{plv}} \in L^2_w(I)$ denote the solution of the PLV~\cite{Miretti-2021}. Then the following four
	statements are equivalent and uniquely characterize $g_{\mathrm{plv}}$.
	
	1) \textbf{Minimum-norm formulation:}
	\begin{equation}
		g_{\mathrm{plv}}
		\;=\;
		\arg\min_{\mathring{g} \in \mathcal{V}_w}
		\|\mathring{g}\|_{w}
		\;=\;
		\mathbb{P}_{\mathcal{V}_w}(0).
		\label{eq:plv-min-norm}
	\end{equation}
	i.e., the minimum-norm
	element in $\mathcal{V}_w$.
	
	2) \textbf{Geometric feasibility formulation:}
	\begin{equation}
		g_{\mathrm{plv}}=\mathcal{V}_w \cap \mathcal{N}_{\bot},
		\label{eq:plv-intersection}
	\end{equation}
	i.e., $g_{\mathrm{plv}}$ is the unique element in the
	intersection set $\mathcal{V}_w \cap \mathcal{N}_{\bot}$.
	
	3) \textbf{Trigonometric polynomial parameterization:} 
	There is a unique vector
	$\bs{b} = [b_0,\ldots,b_{2M-2}]^{\top} \in \mathbb{R}^{2M-1}$
	such that
	\begin{equation}
		\label{eq:plv-trig-form}
		\begin{aligned}
			g_{\mathrm{plv}}(x;\bs{b})
			&:= 
			b_0 
			+ \sum_{m=1}^{M-1} 
			b_m \cos(\kappa_m x)
			+ b_{M-1+m} \sin(\kappa_m x) ,
		\end{aligned}
	\end{equation}
	with the vector being solved by:
	letting $r_{-m} := r_m^{*}$ for $m=-(M-1),\ldots,M-1$,
	\begin{equation}
		\begin{aligned}
			\text{find } \bs{b} \in \mathbb{R}^{2M-1}
			\text{s.t.}\quad
			\langle g_{\mathrm{plv}}(\cdot;\bs{b}),
			e^{\ii \kappa_m (\cdot)} \rangle_w
			= r_m. 
		\end{aligned}
		\label{eq:plv-trig-feas}
	\end{equation}
	
	4) \textbf{Closed-form coefficient representation:} 
	The vector $\bs{b}$ in~\eqref{eq:plv-trig-form} admits the
	closed form
	\begin{equation}
		\bs{b}
		=  \mathbf{G}^{-1} \bs{y},
		\label{eq:plv-closed-form}
	\end{equation}
	where $\bs{y}$ is defined by\footnote{We note that the imaginary of $r_0$ is trivial, i.e., $\Im r_0=0$ (refer to~\eqref{eq:steering}).}
	\begin{equation}
		\label{eq: y-def}
		\bs{y}:=\left[
		\begin{matrix}
			\Re \bs{r}_{0:M-1} \\
			\Im \bs{r}_{1:M-1}
		\end{matrix}
		\right] \in \mathbb R^{2M-1}
	\end{equation} 
	and
	$\mathbf{G}$ satisfies\footnote{Note that matrix $\mathbf G$ slightly differs from the one in~\cite[Proposition~1]{Miretti-2021}; in particular, the latter is rank-deficient, whereas $\mathbf G$ is full-rank in our setting.}
	\begin{equation}
		\mathbf{G}
		:=
		\begin{bmatrix}
			\mathbf{G}_{\Re} & \mathbf{0} \\
			\mathbf{0}       & \mathbf{G}_{\Im}
		\end{bmatrix} \in \mathbb{R}^{(2M-1)\times(2M-1)},
		\label{eq:plv-G-def}
	\end{equation}
	with $\mathbf{G}_{\Re} \in \mathbb{R}^{M\times M}$ and
	$\mathbf{G}_{\Im} \in \mathbb{R}^{(M-1)\times(M-1)}$
	given by
		\[
		\begin{split}
			[\mathbf{G}_{\Re}]_{m,n}
			&= \tfrac{\pi}{2}\big(J_0(\kappa_{m-n}) + J_0(\kappa_{m+n})\big), \\
			[\mathbf{G}_{\Im}]_{m^{\prime}-1,n^{\prime}-1}
			&= \tfrac{\pi}{2}\big(J_0(\kappa_{m^{\prime}-n^{\prime}}) - J_0(\kappa_{m^{\prime}+n^{\prime}})\big),
		\end{split}
		\]
		for $0 \le m,n \le M-1$ and $1 \le m^{\prime},n^{\prime} \le M-1$, respectively, and 
		$J_0(\cdot)$ denotes the Bessel function of the first kind
		of order zero.
		In particular, $\mathbf{G} \succ \bs{0}$ and hence invertible.
	\end{theorem}
	
	\begin{proof}
		See Appendix~\ref{proof: proof-PLV}.
	\end{proof}

\begin{figure}
\centerline{\includegraphics[width=0.4\linewidth]{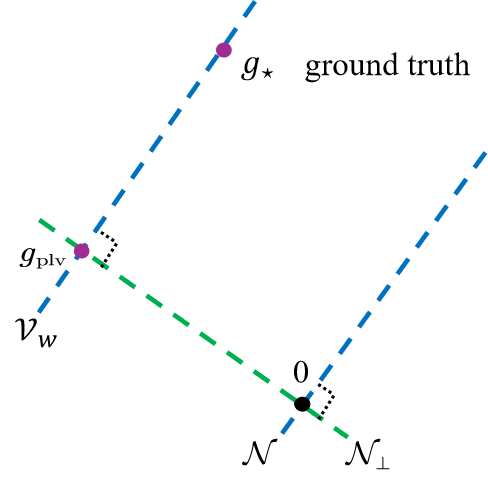}}
\caption{Geometric interpretation of the PLV. The affine set $\mathcal{V}_w$ collects all functions consistent with the covariance constraints, $\mathcal{N}$ denotes the associated null space, and $\mathcal{N}_{\bot}$ is its orthogonal
	complement.}
\label{fig:PLV}
\end{figure}
	
	\begin{remark}[Additional insight into the PLV formulation]
		Building upon the original PLV method in~\cite{Miretti-2021}, Theorem~\ref{prop:PLV-equivalent} makes
		two aspects of the method more explicit.
		\begin{itemize}
			\item 
			\emph{Appearance of trigonometric polynomials}.
			The explicit description of $\mathcal{N}_{\bot}$ shows that
			PLV reconstructs the APS within a fixed-order trigonometric polynomial
			model in the transformed domain $g_{\star} \in L^2_w(I)$. In particular, PLV selects the basis
			$\{1,\cos(\kappa_m x),\sin(\kappa_m x)\}$ for representing APS,
			which clarifies the underlying modeling assumptions on angular
			resolution and on the level of detail that can be captured in the
			reconstructed APS.
			\item
			\emph{Uniqueness}.
			The closed-form relation~\eqref{eq:plv-closed-form} reveals that the PLV
			estimate is linear to the covariance, with a
			matrix $\mathbf{G}$ that depends only on the array geometry and frequencies. Since $\mathbf{G} \succ\bs{0}$, the coefficients $\bs{b}$, and hence
			$g_{\mathrm{plv}}$, are uniquely determined by the covariance
			constraints. This makes the uniqueness of the PLV solution explicit.
		\end{itemize}
	\end{remark}
	
	\textit{Discussion}.
	Theorem~\ref{prop:PLV-equivalent} shows that PLV recovers the APS by projecting
	onto a fixed dimensional trigonometric polynomial subspace
	in the transformed domain $g \in L^2_w(I)$, yielding a unique solution
	that depends linearly on the covariance vector. 
	
	\subsection{Error Decomposition and Resolution}

	\begin{proposition}[Energy identity and reconstruction error]
		\label{prop:plv-energy}
		Let $g_{\mathrm{plv}}$ be the PLV solution as in~\eqref{eq:plv-trig-form}. Then
		\begin{equation}
			\|g_{\mathrm{plv}}\|_w^2 = \bs b^\top \mathbf{G} \bs b = \bs y^\top \mathbf {G}^{-1} \bs y,
		\end{equation}
		where $\bs b= \mathbf{G}^{-1} \bs y$ is given in~\eqref{eq:plv-closed-form}. Moreover, for any $g\in \mathcal{V}_w$,
		\begin{equation}
			\|g\|_w^2 = \|g_{\mathrm{plv}}\|_w^2 + \|g-g_{\mathrm{plv}}\|_w^2.
		\end{equation}
		In particular,
		\begin{equation}
			\|g_{\mathrm{plv}}-g_\star\|_w^2 = \|g_\star\|_w^2 - \bs y^\top \mathbf {G}^{-1} \bs y \ge 0,
		\end{equation}
		with equality if and only if $g_\star\in \mathcal{N}_\perp$.
	\end{proposition}
	\begin{proof}
		See Appendix~\ref{proof: error}.
	\end{proof}
	
	Compared with prior set-theoretic analyses in~\cite{Cavalcante-2018}, our error study differs in two main aspects:
	(i) \cite{Cavalcante-2018} primarily bounds errors of specific {linear functionals} of the APS (task-dependent performance), whereas we bound the \emph{APS reconstruction error} itself in a \emph{weighted} $L^2$ geometry; 
	(ii) We leverage the explicit PLV geometry in the \emph{weighted} Fourier-domain to expose a \emph{resolution/identifiability} limit, while their bounds do not aim to characterize this resolution behavior.
	The following Corollary is a direct consequence of Proposition~\ref{prop:plv-energy}.

	\begin{corollary}[Resolution/Identifiability]
		\label{corollayr: PLV}
		The APS becomes {perfectly recoverable} via the PLV approach if and only if the ground-truth APS $g_{\star} \in \mathcal{N}_{\bot}$, i.e., is represented by the trigonometric polynomials in~\eqref{eq:plv-trig-form} with order up to $2M-2$ (a total of $2M-1$ degrees).
	\end{corollary}

	Corollary~\ref{corollayr: PLV} yields an explicit identifiability (resolution) condition for \emph{continuous} APS recovery, thereby clarifying the fundamental \emph{resolution} limit induced by the finite aperture. To the best of our knowledge, explicit results for such resolution have not appeared in the literature.

    \section{Concluding Remarks}
\label{sec: conclusion}
We studied the recovery of a continuous APS from spatial covariance and revisited the PLV approach through a \emph{weighted} Fourier-domain Hilbert-space viewpoint. This geometric reformulation reveals that PLV operates over an explicit finite-dimensional trigonometric-polynomial model and leads to a closed-form solution characterized by a geometry-dependent, positive-definite matrix, making uniqueness transparent. Moreover, the resulting orthogonal decomposition provides an exact energy/error relation and yields a sharp identifiability (resolution) statement: PLV achieves exact recovery precisely when the true spectrum lies in the identified subspace; otherwise, it returns the minimum-energy spectrum among all covariance-consistent candidates. Future work includes extending the analysis to noisy or finite-sample covariance estimates and to more general array geometries (e.g., planar arrays).

	\appendices

\section{Proof of Lemma~\ref{lem: Affine-sub-space}}
\label{proof: Affine}
\emph{1) Affine structure.}
Let $f_1\in \mathcal V_w$ and $f_2\in \mathcal N$. For every $-M<m<M$, by linearity of the inner product in the first argument,
\[
\langle f_1+f_2, e^{i\kappa_m(\cdot)}\rangle_w
=\langle f_1, e^{i\kappa_m(\cdot)}\rangle_w+\langle f_2, e^{i\kappa_m(\cdot)}\rangle_w
=r_m,
\]
so $f_1+f_2\in \mathcal V_w$. Hence $\mathcal V_w+\mathcal N\subseteq \mathcal V_w$.
Conversely, fix any $f_1\in \mathcal V_w$ and take any $f\in \mathcal V_w$. Then for each $-M<m<M$,
\[
\langle f-f_1, e^{i\kappa_m(\cdot)}\rangle_w
=\langle f, e^{i\kappa_m(\cdot)}\rangle_w-\langle f_1, e^{i\kappa_m(\cdot)}\rangle_w
=0,
\]
so $f-f_1\in \mathcal N$, i.e., $f\in f_1+\mathcal N\subseteq \mathcal V_w+\mathcal N$.
Therefore $\mathcal V_w\subseteq \mathcal V_w+\mathcal N$, and we conclude
\(
\mathcal V_w+\mathcal N=\mathcal V_w,
\)
so $\mathcal V_w$ is an affine subset of $L^2_w(I)$ with direction $\mathcal N$.

\emph{2) Characterization of $\mathcal N_\perp$.}
Define the finite-dimensional subspace
\[
\mathcal S:=\operatorname{span}\bigl(\,1,\ \{\cos(\kappa_m \,\cdot)\}_{m=1}^{M-1},\ \{\sin(\kappa_m \,\cdot)\}_{m=1}^{M-1}\bigr)\subset L^2_w(I).
\]
We claim that $\mathcal N=\mathcal S^\perp$. Indeed, if $f\in \mathcal N$, then in particular
\[
\langle f,1\rangle_w=\langle f,e^{i\kappa_0(\cdot)}\rangle_w=0,
\]
and for each $m=1,\dots,M-1$ we have $\langle f,e^{i\kappa_m(\cdot)}\rangle_w=0$ and
$\langle f,e^{-i\kappa_m(\cdot)}\rangle_w=0$ (since the constraints hold for all $-M<m<M$).
Using the identities
\[
\begin{split}
	\cos(\kappa_m x)=\frac{e^{i\kappa_m x}+e^{-i\kappa_m x}}{2},
	\sin(\kappa_m x)=\frac{e^{i\kappa_m x}-e^{-i\kappa_m x}}{2i},
\end{split}
\]
and linearity of $\langle\cdot,\cdot\rangle_w$ in the second argument, we obtain
\[
\langle f,\cos(\kappa_m(\cdot))\rangle_w
=\frac{1}{2}\big(\langle f,e^{i\kappa_m(\cdot)}\rangle_w+\langle f,e^{-i\kappa_m(\cdot)}\rangle_w\big)=0,
\]
\[
\langle f,\sin(\kappa_m(\cdot))\rangle_w
=\frac{1}{2i}\big(\langle f,e^{i\kappa_m(\cdot)}\rangle_w-\langle f,e^{-i\kappa_m(\cdot)}\rangle_w\big)=0.
\]
Hence $f$ is orthogonal to each generator of $\mathcal S$, i.e., $f\in \mathcal S^\perp$, so $\mathcal N\subseteq \mathcal S^\perp$.
Conversely, if $f\in \mathcal S^\perp$, then $\langle f,1\rangle_w=0$ and for each $m=1,\dots,M-1$,
$\langle f,\cos(\kappa_m(\cdot))\rangle_w=0$ and $\langle f,\sin(\kappa_m(\cdot))\rangle_w=0$; therefore
\[
\langle f,e^{\pm i\kappa_m(\cdot)}\rangle_w
=\langle f,\cos(\kappa_m(\cdot))\rangle_w  \pm i\,\langle f,\sin(\kappa_m(\cdot))\rangle_w
=0,
\]
and also $\langle f,e^{i\kappa_0(\cdot)}\rangle_w=\langle f,1\rangle_w=0$. This is exactly the defining condition of $\mathcal N$,
so $\mathcal S^\perp\subseteq \mathcal N$. Hence $\mathcal N=\mathcal S^\perp$.

Taking orthogonal complements yields
\[
\mathcal N_\perp=\mathcal N^\perp=(\mathcal S^\perp)^\perp=\overline{\mathcal S}.
\]
As $\mathcal S$ is finite-dimensional, it is closed in $L^2_w(I)$, so $\overline{\mathcal S}=\mathcal S$ and thus $\mathcal N_\perp=\mathcal S$.
Equivalently, every $f\in \mathcal N_\perp$ can be written as
\[
\begin{split}
	f(x)&=b_0+\sum_{m=1}^{M-1} b_m \cos(\kappa_m x)+ b_{M-1+m}\sin(\kappa_m x)
	\\
	&:=f(x;b),
\end{split}
\]
for some $b\in\mathbb{R}^{2M-1}$, and conversely every such trigonometric polynomial belongs to $\mathcal N_\perp$.
This proves Lemma~\ref{lem: Affine-sub-space} and the finite-dimensionality claim.

\section{Proof of Theorem~\ref{prop:PLV-equivalent}}
\label{proof: proof-PLV}
We show the equivalence of the four statements in
Theorem~\ref{prop:PLV-equivalent} by linking them pairwise.
Throughout the proof, we use the sets
$\mathcal{V}_w, \mathcal{N}, \mathcal{N}_{\bot}$ defined in
Lemma~\ref{lem: Affine-sub-space}.

\textit{1) $\Leftrightarrow$ 2): minimum norm vs. affine intersection.}
By Lemma~\ref{lem: Affine-sub-space}, the feasible set
$\mathcal{V}_w$ is an affine subset of $L^2_w(I)$ with
direction $\mathcal{N}$, i.e., $\mathcal{V}_w + \mathcal{N}
= \mathcal{V}_w$. In particular, if $g_0 \in \mathcal{V}_w$
is any fixed feasible point, then every $\mathring{g} \in \mathcal{V}_w$
can be written uniquely as
\[
\mathring{g} = g_0 + h, \quad h \in \mathcal{N}.
\]

Assume first that statement~2) holds, i.e.,
$\mathcal{V}_w \cap \mathcal{N}_{\bot} = \{g_{\mathrm{plv}}\}$.
For any $g \in \mathcal{V}_w$, write
$\mathring{g} = g_{\mathrm{plv}} + h$ with $h \in \mathcal{N}$. Since
$g_{\mathrm{plv}} \in \mathcal{N}_{\bot}$ and $h \in \mathcal{N}$,
we have $\langle g_{\mathrm{plv}}, h \rangle_w = 0$ and hence
\[
\|\mathring{g}\|_w^2
= \|g_{\mathrm{plv}} + h\|_w^2
= \|g_{\mathrm{plv}}\|_w^2 + \|h\|_w^2
\ge \|g_{\mathrm{plv}}\|_w^2,
\]
with equality only if $h = 0$, i.e., $\mathring{g} = g_{\mathrm{plv}}$.
Thus $g_{\mathrm{plv}}$ is the unique minimum-norm element in
$\mathcal{V}_w$, which proves statement~1).

Conversely, suppose $g_{\mathrm{plv}}$ is the unique solution of the
minimum-norm problem~\eqref{eq:plv-min-norm}. Decompose
$g_{\mathrm{plv}}$ orthogonally with respect to $\mathcal{N}$ as
\[
g_{\mathrm{plv}} = g_{\perp} + h,
\quad g_{\perp} \in \mathcal{N}_{\bot},\; h \in \mathcal{N}.
\]
Since $\mathcal{V}_w$ is an affine set with direction
$\mathcal{N}$, we have $g_{\perp} = g_{\mathrm{plv}} - h \in
\mathcal{V}_w$. Moreover,
\[
\|g_{\perp}\|_w^2
= \|g_{\mathrm{plv}} - h\|_w^2
= \|g_{\mathrm{plv}}\|_w^2 - \|h\|_w^2
\le \|g_{\mathrm{plv}}\|_w^2,
\]
with strict inequality whenever $h \ne 0$. By the optimality of
$g_{\mathrm{plv}}$, it must hold that $h = 0$, and hence
$g_{\mathrm{plv}} = g_{\perp} \in \mathcal{N}_{\bot}$. Combined with
$g_{\mathrm{plv}} \in \mathcal{V}_w$, this implies
$g_{\mathrm{plv}} \in \mathcal{V}_w \cap \mathcal{N}_{\bot}$.
Finally, if there were two distinct elements
$g_1, g_2 \in \mathcal{V}_w \cap \mathcal{N}_{\bot}$, then
$g_1 - g_2 \in \mathcal{N} \cap \mathcal{N}_{\bot} = \{0\}$, so
$g_1 = g_2$. Hence $\mathcal{V}_w \cap \mathcal{N}_{\bot}$ is
a singleton, and statement~2) holds. This proves $1) \Leftrightarrow 2)$.

\textit{2) $\Leftrightarrow$ 3): geometric formulation vs. trigonometric parameterization.}
Lemma~\ref{lem: Affine-sub-space} further shows that
$\mathcal{N}_{\bot}$ is a $(2M-1)$-dimensional subspace spanned by
the trigonometric basis
\[
\big\{ 1,\,
\cos(\kappa_m x),\,
\sin(\kappa_m x)
\;\big|\;
m = 1,\ldots,M-1 \big\},
\]
so every $g \in \mathcal{N}_{\bot}$ admits a unique representation of
the form~\eqref{eq:plv-trig-form} with some
$\bs{b} \in \mathbb{R}^{2M-1}$. In particular, any
$g_{\mathrm{plv}} \in \mathcal{V}_w \cap \mathcal{N}_{\bot}$
must be representable as in~\eqref{eq:plv-trig-form}, and the
constraints defining $\mathcal{V}_w$ translate exactly into the
feasibility conditions~\eqref{eq:plv-trig-feas}. This yields
statement~3).

Conversely, if $g_{\mathrm{plv}}$ admits the representation
\eqref{eq:plv-trig-form} with a vector $\bs{b}$ that
satisfies~\eqref{eq:plv-trig-feas}, then by construction
$g_{\mathrm{plv}} \in \mathcal{N}_{\bot}$ and
$\langle g_{\mathrm{plv}}, e^{\ii\kappa_m(\cdot)} \rangle_w = r_m$
for all $-M < m < M$, i.e., $g_{\mathrm{plv}} \in \mathcal{V}_w$.
Uniqueness of $\bs{b}$ (and hence of $g_{\mathrm{plv}}$) follows from
the linear independence of the trigonometric basis. Therefore
$\mathcal{V}_w \cap \mathcal{N}_{\bot}
= \{g_{\mathrm{plv}}\}$, and statement 2) holds. This proves
$2) \Leftrightarrow 3)$.

\textit{3) $\Leftrightarrow$ 4): trigonometric form vs. closed-form coefficients.}
Substituting the expansion~\eqref{eq:plv-trig-form} into the
covariance constraints
$\langle g_{\mathrm{plv}}, e^{\ii\kappa_m(\cdot)} \rangle_w = r_m$,
$-M < m < M$, and separating real and imaginary parts yield a linear
system of the form
\begin{equation}
	\label{eq:plv-linear-system}
	\mathbf{G}\,\bs{b} = \bs{y},
\end{equation}
where $\bs{y}$ is defined
in~\eqref{eq: y-def} and $\mathbf{G}$ is the block-diagonal matrix in~\eqref{eq:plv-G-def}. The entries of $\mathbf{G}_{\Re}$ and $\mathbf{G}_{\Im}$
can be written as
\[
\begin{split}
	[\mathbf{G}_{\Re}]_{m,n}
	&= \langle \cos(\kappa_m \cdot), \cos(\kappa_n \cdot) \rangle_w,
	\\
	[\mathbf{G}_{\Im}]_{m-1,n-1}
	&= \langle \sin(\kappa_m \cdot), \sin(\kappa_n \cdot) \rangle_w.
\end{split}
\]
with (via the identity$\int_{0}^{\pi} \cos(z\cos\theta)\dd \theta=\pi J_0(z)$)
\[
\begin{split}
	\langle \cos(\kappa_m \cdot), \cos(\kappa_n \cdot) \rangle_w &= \tfrac{\pi}{2}\big(J_0(\kappa_{m-n}) + J_0(\kappa_{m+n})\big),
	\\
	\langle \sin(\kappa_m \cdot), \sin(\kappa_n \cdot) \rangle_w &= \tfrac{\pi}{2}\big(J_0(\kappa_{m^{\prime}-n^{\prime}}) - J_0(\kappa_{m^{\prime}+n^{\prime}})\big).
\end{split}
\]
For any $\bs{c} \in \mathbb{R}^{M}$, we have
\[
\bs{c}^{\top} \mathbf{G}_{\Re} \bs{c}
= \big\| \sum_{m=0}^{M-1} c_m \cos(\kappa_m \cdot) \big\|_w^2
\ge 0,
\]
with equality if and only if
$\sum_{m=0}^{M-1} c_m \cos(\kappa_m x) = 0$ for all $x \in I$,
which by linear independence of the cosine functions implies
$\mathbf{c} = \mathbf{0}$. Hence $\mathbf{G}_{\Re}$ is symmetric
positive definite. An analogous argument shows that
$\mathbf{G}_{\Im}$ is symmetric positive definite as well. Thus
$\mathbf{G}$ is positive definite and invertible, and
\eqref{eq:plv-linear-system} has the unique solution
$\bs{b} =  \mathbf{G}^{-1} \bs{y}$, which gives the
closed-form expression~\eqref{eq:plv-closed-form}. This proves
statement~4) given statement~3).


Conversely, if $\bs{b} = \mathbf{G}^{-1} \bs{y}$, then
\eqref{eq:plv-linear-system} holds, and thus 
$g_{\mathrm{plv}}(\cdot;\bs{b})$ of the
form~\eqref{eq:plv-trig-form} satisfies covariance constraints in $\mathcal{V}_w$
and belongs to $\mathcal{N}_{\bot}$. This recovers statement~3),
and hence $3) \Leftrightarrow 4)$.

Combining the implications $1) \Leftrightarrow 2)$,
$2) \Leftrightarrow 3)$, and $3) \Leftrightarrow 4)$ yields the
equivalence of all four characterizations.

\section{Proof of Proposition~\ref{corollayr: PLV}}
\label{proof: error}
For any $x \in I=[-1,1]$, denote by 
\[
\bm \phi(x) :=
[ 1, \ldots, \cos(\kappa_m x),\ldots, \sin(\kappa_m x),\ldots\,]^\top \in \mathbb R^{2M-1},
\]
and recall Theorem~\ref{prop:PLV-equivalent} that
\(
g_{\mathrm{plv}}(x)=\bm b^T\bm\phi(x)
\)
with $\bs b=\mathbf{G}^{-1}\bs y$, where the matrix is
\(
\mathbf{G}=\int_I \bm\phi(x)\bm\phi(x)^\top w(x)\,dx.
\)
Then
\begin{align*}
	\|g_{\mathrm{plv}}\|_w^2
	&=\langle g_{\mathrm{plv}},g_{\mathrm{plv}}\rangle_w
	\\
	&=\bm b^T\!\big(\int_I \bm\phi(x)\bm\phi(x)^T w(x)\,dx\big)\!\bm b
	\\
	&=\bm b^T \mathbf{G}\bm b
	=\bm y^T \mathbf{G}^{-1}\bm y,
\end{align*}
where the last equality uses $\bs b=\mathbf{G}^{-1}\bs y$ and $\mathbf{G}=\mathbf{G}^\top \succ \bs 0$.

Next, by Theorem~\ref{prop:PLV-equivalent}, $g_{\mathrm{plv}}=\mathbb P_{\mathcal{V}_w}(0)$, i.e., it minimizes $\|g\|_w$ over the affine set $\mathcal{V}_w=g_{\mathrm{plv}}+\mathcal{N}$.
Hence for any $h\in \mathcal{N}$, the function $f(t):=\|g_{\mathrm{plv}}+th\|_w^2$ is minimized at $t=0$, which implies
\(
0=f'(0)=2\langle g_{\mathrm{plv}},h\rangle_w
\)
for all $h\in N$, i.e., $g_{\mathrm{plv}}\perp \mathcal{N}$.
Therefore, for any $g\in \mathcal{V}_w$ we can write $g=g_{\mathrm{plv}}+h$ with $h\in \mathcal{N}$, and
\[
\begin{split}
\|g\|_w^2
&=\|g_{\mathrm{plv}}+h\|_w^2
=\|g_{\mathrm{plv}}\|_w^2+\|h\|_w^2
\\
&=\|g_{\mathrm{plv}}\|_w^2+\|g-g_{\mathrm{plv}}\|_w^2.
\end{split}
\]
In particular, for $g_\star\in V_r$,
\(
\|g_{\mathrm{plv}}-g_\star\|_w^2
=\|g_\star\|_w^2-\bm y^T \mathbf{G}^{-1}\bm y.
\)
This completes the proof.

\bibliographystyle{IEEEtran}
\bibliography{IEEEabrv,refs1}

\end{document}